\documentclass[twoside]{article}

\usepackage[preprint]{aistats2026}

\usepackage{hyperref}
\usepackage{url}
\usepackage[utf8]{inputenc} 
\usepackage[T1]{fontenc}    
\usepackage{booktabs}       
\usepackage{amsfonts}       
\usepackage{nicefrac}       
\usepackage{microtype}      
\usepackage{xcolor}         
\usepackage{graphicx}
\usepackage{amsmath}
\usepackage{physics}
\usepackage{amsthm}
\usepackage{algorithm2e}
\usepackage{algorithmic}

\theoremstyle{plain}
\newtheorem*{theorem*}{Theorem}
\newtheorem{theorem}{Theorem}

\newcommand{\x}{\boldsymbol{x}}
\newcommand{\U}{\boldsymbol{u}}
\newcommand{\Z}{\boldsymbol{z}}

\usepackage[round]{natbib}
\renewcommand{\bibname}{References}
\renewcommand{\bibsection}{\subsubsection*{\bibname}}
\PassOptionsToPackage{hyphens}{url} 
\usepackage{xurl}     
\begin{document}

%

%

\twocolumn[

\aistatstitle{Faster parallel MCMC: Metropolis adjustment is best served warm}

\aistatsauthor{Jakob Robnik \And Uroš Seljak}

\aistatsaddress{ Physics Department,\\
University of California at Berkeley,\\
Berkeley, CA 94720, USA \And  Physics Department,\\
University of California at Berkeley\\
and Lawrence Berkeley National Laboratory, \\
Berkeley, CA 94720, USA} ]

\begin{abstract}
    Despite the enormous success of Hamiltonian Monte Carlo and related Markov Chain Monte Carlo (MCMC) methods, sampling often still represents the computational bottleneck in scientific applications. Availability of parallel resources can significantly speed up MCMC inference by running a large number of chains in parallel, each collecting a single sample. However, the parallel approach converges slowly  if the chains are not initialized close to the target distribution (cold start). Theoretically this can be resolved by initially running MCMC without Metropolis-Hastings adjustment to quickly converge to the vicinity of the target distribution and then turn on adjustment to achieve fine convergence. However, no practical scheme uses this strategy, due to the difficulty of automatically selecting the step size during the unadjusted phase. We here develop Late Adjusted Parallel Sampler (LAPS), which is precisely such a scheme and is applicable out of the box, all the hyperparameters are selected automatically. LAPS takes advantage of ensemble-based hyperparameter adaptation to estimate the bias at each iteration and converts it to the appropriate step size. We show that LAPS consistently and significantly outperforms ensemble adjusted methods such as MEADS or ChESS and the optimization-based initializer Pathfinder on a variety of standard benchmark problems. LAPS typically achieves two orders of magnitude lower wall-clock time than the corresponding sequential algorithms such as NUTS.
\end{abstract}

\section{Introduction}

Markov Chain Monte Carlo (MCMC) methods \citep{metropolis_equation_1953} are a workhorse of Bayesian statistical inference, with applications ranging from
from economics and social science \cite{gelman_bayesian_1995}, to high energy physics \citep{duane_hybrid_1987} and machine learning \citep{neal_bayesian_1996}.
They are also widely employed for performing high-dimensional integrals in quantum physics \citep{gattringer_quantum_2010}, molecular chemistry \citep{leimkuhler_molecular_2015}, statistical mechanics, and in many other fields.
In all of these applications, the goal of MCMC is to estimate expectation values $\mathbb{E}_{p}[f] = \int f(\x) p(\x) d\x$, where $p(\x)$ is a probability density distribution of parameters $\x \in \mathbb{R}^d$. This is achieved by generating $M$ samples $\{ \x^{m} \}_{m=1}^M$ from the probability density distribution $p(\x)$ and to estimate expectation values of interest as $\mathbb{E}_p[f(\x)] \approx \frac{1}{M} \sum_{m = 1}^M f(\x^{m})$. 
We focus on situation where a (smooth) gradient of the target distribution $\nabla \log p(\x)$ is available, either analytically, or via automatic differentiation \citep{griewank_evaluating_2008}. This is often the case in fields like Bayesian statistics \citep{strumbelj_past_2024, carpenter_stan_2017}, machine learning \citep{baydin_automatic_2018}, lattice quantum problems \citep{gattringer_quantum_2010}, and cosmology \citep{campagne_jax-cosmo_2023}.

Traditionally, MCMC is run in the sequential mode, where one (or a few) long chains are used to generate samples \citep{margossian_for_2024}. The chain first undergoes a burn-in phase where it converges from the initial state to the stationary distribution. After the burn-in, the samples are collected and the expectation value error decays as the square root of the effective sample size, which is the total number of samples, reduced by the sample correlation length \citep{gelman_bayesian_1995}.
However, the sampling phase can be very long in wall-clock time for expensive likelihoods or complicated geometries, which is not practical and can even be prohibitive for a typical scientific workflow.

\paragraph{Parallel sampling.}
Following the recent computing advances enabling accessibility of parallel resources, such as GPUs, TPUs and CPU clusters, a new paradigm has emerged \citep{sountsov_running_2024, hoffman_tuning-free_2022, lao_tfpmcmc_2020}: instead of running a few very long chains, where each chain produces many samples, run a large number of chains in parallel, such that each chain only collects one sample. This strategy is known as ensemble MCMC \citep{hoffman_tuning-free_2022, goodman_ensemble_2010} or many-short-chains regime \citep{margossian_nested_2024}.
In the ensemble approach, it is essentially only the burn-in phase that matters; after it finishes, the distribution of the chain positions is already the stationary distribution, and the error is only determined by the number of chains. However, this approach is sensitive to the initialization of the chains, and can converge slowly if the chains are not initialized close to the target distribution \citep{durmus_asymptotic_2023}.

\paragraph{Late adjustment.} It is theoretically well-understood \citep{altschuler_faster_2024, durmus_asymptotic_2023} that Metropolis-Hastings (MH) adjusted MCMC excels when the initial distribution is close to the target distribution (warm start), but suffers if the initial distribution is far from the stationary distribution (cold start), because the acceptance probability is significantly lower away from the typical set \citep{durmus_asymptotic_2023}. Unadjusted methods (no MH accept-reject step) are complementary in this regard: they are fast to converge from a cold start to the vicinity of the stationary distribution, but struggle to achieve fine convergence, since 
they are biased and need to control the bias by the small step size \citep{bou-rabee_convergence_2023, bou-rabee_mixing_2023, camrud_second_2024}. One can take the best of both worlds by using an unadjusted method to converge from the cold start to the warm start and then use an adjusted method to achieve fine convergence \citep{durmus_asymptotic_2023, altschuler_faster_2024}. This method gives theoretically best mixing time to high accuracy from a cold start to date. However, these developments are theoretical and assume global information about the target that is not available in practice, so it is not immediately clear how to implement this algorithm and in particular, how to select its step size during the unadjusted phase. This will be the focus of the present work. 

\paragraph{Our contributions.} We develop a first practical parallel algorithm that initially performs unadjusted sampling and then uses MH adjustment to achieve fine convergence. The resulting algorithm, termed Late-Adjusted Parallel Sampler (LAPS) can be applied out of the box, the user only needs to provide the target density and the initialization of the chains. We show on standard benchmark problems that it significantly outperforms other turnkey ensemble samplers such as ChESS-HMC \citep{hoffman_adaptive-mcmc_2021} or MEADS \citep{hoffman_tuning-free_2022}, in some cases by more than an order of magnitude. It also achieves two orders of magnitude faster convergence than sequential algorithms such as NUTS \citep{hoffman_no-u-turn_2014} or MCLMC \citep{robnik_microcanonical_2024}. The algorithm is implemented in JAX \citep{jax}, specifically in Blackjax \citep{blackjax} and comes with tutorials on how to use it\footnote{Blackjax: \url{https://blackjax-devs.github.io/blackjax/examples/quickstart.html}.}\footnote{Code for reproducing the results is also public: \url{https://github.com/reubenharry/sampler-benchmarks}}. 
Blackjax can easily be combined with models written in probabilistic programming languages such as Numpyro \citep{phan_composable_2019} or Tensorflow Probability \citep{tfp}.

\section{Background and related work}

\paragraph{Parallel MCMC.} The strategy is to evolve an ensemble of $M$ chains.
Let $(\x^m_t, \U^m_t)$ be the state of $m$-th chain at iteration $t$, where $\x^m_t$ is the position in the parameter space and $\U^m_t$ is the velocity. 
Further let $\rho_t(\x, \U)$ be the underlying ensemble probability density distribution, that is, $(\x^m_t, \U^m_t)$ are exact samples from $\rho_t(\x, \U)$. 
Each chain is an independent Markov process,
\begin{equation}
    (\x_{t+1}^m, \U_{t+1}^m) = \varphi(\x_{t}^m, \U_{t}^m \vert \beta_t),
\end{equation}
where $\varphi$ is an update kernel and $\beta_t$ are hyperparameters of the kernel. These updates can be applied in parallel due to independence.

\paragraph{Choice of the kernel.} State-of-the-art kernels for sampling from differentiable distributions include Hamiltonian Monte Carlo (HMC; \citet{duane_hybrid_1987, neal_mcmc_2011}) and underdamped Langevin Monte Carlo (LMC; \citet{leimkuhler_molecular_2015}). In both, the velocity (or equivalently momentum) is enforced to have a standard Gaussian marginal. Recently, Microcanonical Langevin Monte Carlo (MCLMC; \citet{robnik_microcanonical_2024, robnik_fluctuation_2024}) has been proposed, where the norm of the velocity is instead kept fixed. The proposed benefit is increased stability to large gradients and faster convergence, due to the more deterministic dynamics. Indeed, large improvements over the state-of-the art have also been observed in practice \citep{simon-onfroy_benchmarking_2025, crespi_flinch_2025, sommer_microcanonical_2024}. Speed of convergence is of particular interest in the present work, so we will adopt MCLMC, rather than the more standard HMC or LMC.  We perform an ablation study in Appendix \ref{sec: ablations}.
In ideal, continuous-time MCLMC, the kernel would be chosen in such a way that $(\x_s, \U_s) = \varphi(\x_0, \U_0 \vert s, \eta)$ exactly solves the following stochastic differential equation (SDE):
\begin{align} \label{eq: SDE}
    d \x_s &= \U_s ds\\ \nonumber
    d \U_s &= (I - \U_s \U_s^T) (\nabla \log p(\x_s) / (d-1) ds + \eta d \boldsymbol{W}), 
\end{align}
with initial condition $(\x_0, \U_0)$. The ensemble distribution $\rho_t$ then provably converges to the target distribution \citep{robnik_fluctuation_2024}. Here, $\boldsymbol{W}$ is the Wiener process \citep{oksendal_stochastic_2003} and $\eta$ is a hyperparameter that determines the strength of stochasticity. In the discretized version of Equation \eqref{eq: SDE} (see Equation \eqref{eq: O update}) we will parameterize the strength of stochasticity by a parameter $L$ rather than $\eta$, such that $\eta = \sqrt{2 / L d}$ in the limit of small step size. This is a convenient parameterization because it makes $L$ directly comparable with with the trajectory length in HMC \citep{robnik_fluctuation_2024}. 

\paragraph{Metropolis test.} In practice, we have to numerically construct a kernel that only approximately solves Equation \eqref{eq: SDE} by a splitting scheme \citep{robnik_fluctuation_2024}, as reviewed in Appendix \ref{sec: integrator}. One then has two ways of dealing with the approximation error:
\begin{itemize}
    \item \textbf{Unadjusted:} use this approximate solution as a kernel $\varphi(\cdot \vert \epsilon, L)$ and ensure that the step size $\epsilon$ is sufficiently small to make the approximation error tolerable. 
    
    \item \textbf{Adjusted:} construct the kernel by performing $N$ steps with the approximate solution $\varphi(\cdot \vert \epsilon, L)$ and wrap them by a Metropolis-Hastings accept/reject step, followed by refreshing the velocity from its marginal distribution \citep{robnik_metropolis_2025}. This construction is reviewed in Appendix \ref{sec: MAMS}. It guarantees that the ensemble exactly converges to the target distribution.
\end{itemize}

Unadjusted approach is faster to converge from the cold start to the warm start, while the adjusted approach is faster to achieve fine convergence. LAPS combines the best of both worlds, by switching on adjustment once the unadjusted sampling stops making sufficient progress. 

\paragraph{Hyperparameter adaptation.} In either case, the kernel has two hyperparameters, the step size $\epsilon$ and the effective trajectory length $L$. In adjusted phase, the number of steps per proposal $N$ is an additional hyperparameter. To tune those parameters automatically we will exploit Ensemble Chain Adaptation (ECA; \citet{gilks_adaptive_1994}), the idea that the ensemble averages at iteration $t$, 
\begin{equation} \label{eq: ECA}
    \mathbb{E}_{\rho_t}[f] \approx \frac{1}{M} \sum_{m = 1}^M f(\x^m_t),
\end{equation}
can be used to inform the hyperparameters in iteration $t+1$. This breaks the complete independence of the chains and introduces an asymptotic bias, which is small if the number of chains is large. To eliminate this bias one can either adopt a scheme from \citep{hoffman_tuning-free_2022} or freeze the hyperparameter tuning at some point, as we do here.
Several ECA algorithms have been developed for HMC or LMC, among these ChEES-HMC \citep{hoffman_adaptive-mcmc_2021}, SNAPER \citep{sountsov_focusing_2022}, MEADS \citep{hoffman_tuning-free_2022} and MALT \citep{riou-durand_adaptive_2023}. Another option is NUTS \citep{hoffman_no-u-turn_2014} run in a parallel mode, where each chain is run independently. These samplers are Metropolis adjusted during the entire process and are typically initializd with an optimizer. A popular choice is the Pathfinder algorithm \citep{zhang_pathfinder_2022}, an L-BFGS-based optimizer that is designed to terminate before the optimizer would start moving away from the typical set and towards the mode of the distribution. 
Our main novelty is that we initialize with an unadjusted sampler. 


\section{Unadjusted initialization} \label{sec: u}

We here develop a scheme to determine the hyperparameters during the unadjusted phase of the algorithm. We emphasize that the purpose of this stage is not to guarantee convergence to the target distribution, but only to approach it as fast as possible.

\subsection{Step size} 
Stationary distribution $p_{\epsilon}$ of the unadjusted kernel $\varphi(\cdot \vert \epsilon, L)$ differs from the target distribution $p$. This is called the asymptotic bias. With larger step size we are moving towards the stationary distribution faster, but also the asymptotic bias is larger. We expect the optimal step size to be time-dependent: when we are further away from stationarity, the asymptotic bias is less important, and the advantage of making larger steps prevails. Closer to the stationary distribution, smaller step size will have to be used to achieve fine convergence. 
In Appendix \ref{sec: schedule} we provide theoretical justification for this intuition and demonstrate that the optimal step size schedule should be such that the asymptotic bias $D(p, p_{\epsilon})$ is smaller than current total bias $D(p, \rho_t)$ by some a time-independent factor, $D(p, p_{\epsilon_t}) = C D(p, \rho_t)$. Here $D$ is some metric on the space of distributions. Adopting this prescription in practice presents several challenges:

\paragraph{Total bias.} Computing distances such as Wasserstein distance or the total variation distance is hard, even when the exact samples are given, let alone in the present situation, where samples from $p(\x)$ are the final goal of the algorithm. Instead, we will construct a proxy for the distance, which is based on a certain class of summary statistics, for which the ground truth is known and is independent of $p(\x)$. We will examine the equipartition condition \citep{rover_partition_2023}: the matrix 
\begin{equation}
    V_{i j}(p, \rho_t) = \mathbb{E}_{\rho_t}[- (x_i - \mathbb{E}_{\rho_t}[x_i]) \partial_j \log p(\x) ]
\end{equation}
equals the identity when $\rho_t = p$, as can be seen by the integration by parts:
\begin{align} \nonumber
     V_{ij}(p, p) &= -\int (x_i - \mathbb{E}_p[x_i]) \partial_j [\log p(\x)] p(\x) d\x \\ \nonumber
     &= -\int (x_i - \mathbb{E}_p[x_i]) \partial_j p(\x) d\x \\
     &= \delta_{i j} \int p(\x) d \x = \delta_{i j}.
\end{align}
We use this fact to construct
\begin{equation} \label{eq: equipartition}
    \widetilde{D}(p, \rho_t) = \frac{1}{d} \Tr{(I - V(p, \rho_t)) (I - V(p, \rho_t))^T}.
\end{equation}
$\widetilde{D}$ is not a metric on the space of distributions, because it is only based on certain summary statistics which do not define the full distribution uniquely. Nevertheless, in Appendix \ref{sec: equipartition} we show that it is a divergence on the space of zero-mean Gaussian distributions. 
Crucially, $\widetilde{D}(\rho_t, p)$ can be calculated using the ECA formalism \eqref{eq: ECA} with negligible computational overhead, as shown in Appendix \ref{sec: equipartition}.

\paragraph{Proportionality constant.} The optimal constant $0< C < 1$ is not known and can depend on the problem at hand. In this work we will fix it to $C = 0.025$. We perform ablation studies in Appendix \ref{sec: ablations}, which demonstrate that this value works well for the problems considered here and the performance is stable for a large range of other values.

\paragraph{Asymptotic bias.} The challenge is how to estimate asymptotic bias $D(p, p_{\epsilon})$ without knowing the ground truth or even converging to the asymptotic regime. Dynamics, such as HMC, LMC and MCLMC have a quantity called the energy which does not change with time for the exact solution. A step with the approximate dynamics, however, changes the energy by amount $\Delta(\x, \U \vert p, \epsilon)$. Exact form of $\Delta$ is given in Appendix \ref{sec: integrator}. Crucially, it only depends on $\log p(\x)$ and $\nabla \log p(\x)$, which are both a side product of the kernel update computation, so the energy change can be computed with negligible computational overhead. Of interest will be the variance of the energy change over the ensemble, specifically the quantity Energy Error Variance per Dimension (EEVPD):
\begin{equation}
    \mathrm{EEVPD}(\rho \vert p, \epsilon) = \mathrm{Var}_{\rho(\x, \U)}[\Delta(\x, \U \vert p, \epsilon)] / d,
\end{equation}
where the variance is $\mathrm{Var}_{\rho}[f] = \mathbb{E}_{\rho}[(f - \mathbb{E}_{\rho}[f])^2]$.
EEVPD is a useful quantity, because it provides an upper bound for the asymptotic bias \citep{robnik_black-box_2024} through 
\begin{equation} \label{eq: bias bound}
    \widetilde{D}(p, p_{\epsilon}) \leq F^{-1}(\mathrm{EEVPD}(p_{\epsilon} \vert p, \epsilon)),
\end{equation}
where $F(\widetilde{D}) = 4 \widetilde{D}^{3/2} / (1 + \widetilde{D}^{1/2})^2$ is a monotonically increasing function. This bound has been demonstrated rigorously for Gaussian distributions, but also numerically shown to typically hold for non-Gaussian distributions. 
\citep{robnik_black-box_2024} also show that EEVPD scales roughly as $\mathrm{EEVPD} \propto \epsilon^6$.

In summary: optimal step size schedule should maintain an asymptotic bias that is smaller than the total bias by some fixed constant. We estimate the total bias by the equipartiton condition and the asymptotic bias by the energy error and then select a step size according to the optimal schedule condition.
This results in the following scheme for determining the step size $\epsilon_{t+1}$, given the ensemble at iteration $\rho_t$.
\begin{enumerate}
    \item Calculate $\widetilde{D}(p, \rho_t)$ from Equation \eqref{eq: equipartition}.
    \item Convert it to the desired asymptotic bias by $\widetilde{D}(p, p_{\epsilon_t}) = C \widetilde{D}(p, \rho_t)$.
    \item Convert the desired asymptotic bias to the desired
    $\mathrm{EEVPD} = F(\widetilde{D}(p, p_{\epsilon_t}))$.
    \item Estimate $EEVPD(\rho_t \vert p, \epsilon_t)$ associated with the current step. Equation \eqref{eq: ECA} is used to compute the variance.
    \item Using the approximate scaling of EEVPD with the step size determine the step size that should be used in the next iteration to ensure the desired EEVPD,
    $\epsilon_{t + 1} = \epsilon_t \big(\mathrm{EEVPD} / \mathrm{EEVPD}(\rho_t \vert p, \epsilon_t))^{1/6}$.
\end{enumerate}

Figure \ref{fig:GermanCredit} shows that with this procedure (i) EEVPD succesfully traces its desired value (ii) step size decays which translates to a rapid decay in bias.

\begin{figure*}
    \centering
    \includegraphics[width=\linewidth]{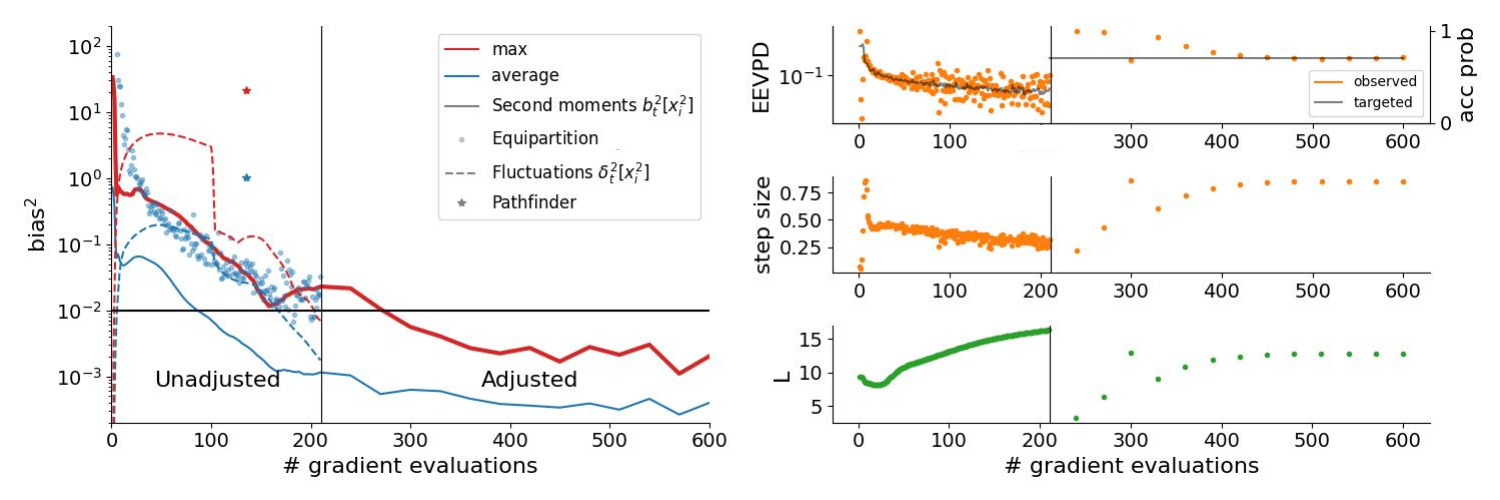}
    \caption{
    LAPS convergence for the German Credit problem. Left panel: bias of the second moments relative to the ground truth, $b^2_{\mathrm{max}}$ (red) and $b^2_{\mathrm{avg}}$ (blue), equipartition loss (blue dots) and the second-moment fluctuations (dotted lines) are shown as a function of the number of gradient evaluations. The latter two can be estimated when the ground truth is not known. In this example, they reflect the actual bias quite well. The switch from the unadjusted to the adjusted phase is shown by a vertical line. In this case, the unadjusted phase already comes very close to the desired accuracy, but adjusted method further improves on that. Note that the final ensemble is not biased, the residual error is caused by the finite number of chains. Pathfinder's second-moment bias is shown with stars.  
    Right panels: convergence of hyperparameters step size and L is shown in the lower two panels. Step size is tuned based on the quantities derived from the energy error, shown in the upper panel. EEVPD is used in the unadjusted phase, acceptance rate in the adjusted phase. Targeted EEVPD and acceptance rate are shown in grey. Acceptance rate is fixed, while targeted EEVPD is adaptively changed, based on the current bias estimate from the equipartition loss. Note that the step size in the adjusted phase converges within a few steps, thanks to the bisection method.
    }
    \label{fig:GermanCredit}
\end{figure*}

\subsection{Momentum decoherence scale} 

Noise added to the dynamics can speed-up mixing by ensuring ergodicity. However, the noise strength should not be too large, so that the dynamics can efficiently explore the posterior. Optimal time scale before the momentum direction decoheres was shown to be comparable to the typical scale of the posterior \citet{robnik_microcanonical_2024}, suggesting the prescription 
\begin{equation} \label{eq: L}
    L_t = \alpha \bigg( \sum_{i = 1}^d \mathrm{Var}_{\rho_t}[ x_i ] \bigg)^{1/2}.
\end{equation}
In \citet{robnik_microcanonical_2024}, $\alpha = 1$ was used as default, while here we find that a larger value $\alpha = 2$ works better, corresponding to the more deterministic dynamics. See the ablation studies in Appendix \ref{sec: ablations}. 

\section{Switch to adjustment} \label{sec: u2a}

Unadjusted sampler with a decaying step size is in principle capable of converging to the exact target distribution. However, we had to resort to several approximations in the previous section, because it is in practice difficult to have control over the asymptotic bias. But even compared to a perfect unadjusted step size schedule, the adjusted are faster to achieve fine convergence. 


We will take a pragmatic approach in determining when to turn on adjustment: we will monitor summary statistics and determine when their values are no longer changing significantly, signaling that the chains are close to stationarity. We define  a relative fluctuation of a summary statistic $\mathbb{E}[f]$ as
\begin{equation} \label{eq: relative fluctuations}
    \delta_t[f] = \sigma_t[f] / \mu_t[f],
\end{equation}
where average and standard deviation are computed in a moving window of length $T$:
\begin{align} 
    \mu_t[f] &= \frac{1}{T}\sum_{s= t-T}^T \mathbb{E}_{\rho_s}[f] \\ \nonumber
    \sigma^2_t[f] &= \frac{1}{T-1}\sum_{s= t-T}^T (\mathbb{E}_{\rho_s}[f] - \mu_t[f])^2 .
\end{align}
These can be implemented in a memory-efficient way by performing an online average, so there is no need to store $T$ states of the chain on each device.

If chains were stationary and in the target distribution, we would expect $\delta[f] = \mathrm{Var}_p[f]^{1/2}  M^{-1/2} / \mathbb{E}_p[f] \sim M^{-1/2}$, so it should be very small if the ensemble is large. A small value of $\delta$ thus indicates that we are approaching equilibrium.
For observables, we will take $f(\x) = x_i^2$ and terminate the unadjusted phase when all $\delta_t[x_i^2]$ fall below a certain threshold. We fix the threshold to $0.01$ and the moving window size to the $20 \%$ of the total sampling time.

\textbf{Diagonal preconditioning} helps to reduce the condition number if the specified target has parameters with different posterior widths, and is thus a common practice in MCMC \citep{sountsov_focusing_2022}.
At the end of the unadjusted phase, we perform diagonal preconditioning by a coordinate transformation $y_i(\x) = x_i / \mathrm{Var}_{\rho_t}[x_i]^{1/2}$.

\section{Adjusted sampling} \label{sec: a}

We now turn on Metropolis adjustment and adapt the hyperparameters to the new regime using ECA. 

\subsection{Step size} 
In the adjusted phase we tune the step size to achieve a desired average acceptance rate \citep{neal_mcmc_2011}. 
Because the number of chains is large, the average acceptance rate estimate has a very small noise. As a consequence, stochastic optimization routines, such as dual averaging \citep{hoffman_no-u-turn_2014} are not needed. Instead, the acceptance rate as a function of step size $a(\epsilon)$ is practically a monotonically decreasing function and we can simply use bisection to find the root of $a(\epsilon) - a_{\mathrm{targeted}}$. First, we double (halve) the step size in each iteration if the initial step size is too small (too large) until we have a bracketing interval, i.e. two values of $\epsilon$ that bracket the root. We then run bisection until the acceptance rate equals the desired acceptance rate, up to some tolerance ($3\%$ by default). Thereafter the hyperparameters are fixed, which eliminates the bias caused by ECA (which is in typically unobservably small).
Figure \ref{fig:GermanCredit} illustrates that this procedure converges within a few steps. In our experiments, bisection was faster and more stable than dual averaging from \citet{hoffman_no-u-turn_2014}.
Extending the argument from \citet{beskos_optimal_2013} and \citet{neal_mcmc_2011}, \citet{betancourt_optimizing_2015} argues that the optimal acceptance rate for product targets with HMC is in the range $60\%-80 \%$ for second-order integrators and $80\%-87\%$ for fourth-order integrators. The same argument applies for MCLMC \citep{robnik_metropolis_2025}. Slightly larger values of acceptance rate are typically more stable numerically \citep{phan_composable_2019}, so we default the target acceptance rate to $70 \%$ for second-order integrators and $90 \%$ for fourth-order integrators. Ablation studies in Appendix \ref{sec: ablations} demonstrate that these choices have marginal impact on the performance.

\subsection{Momentum decoherence scale} 
We do not adapt the trajectory length in this phase, because the phase is typically quite short. We fix 15 integration steps per trajectory, i.e. trajectory length in between full refreshments is $L_{\mathrm{full}} =  15 \epsilon$. We set the $L$ in partial refreshments from Equation \eqref{eq: O update} to $L = 1.25 \, L_{\mathrm{full}}$, as recommended in \citet{riou-durand_metropolis_2023, riou-durand_adaptive_2023}.

\section{Experiments} \label{sec: experiments}

We test LAPS on several benchmark problems, taken from the Inference gym \citep{inferencegym2020}:
\begin{itemize}
    \item Banana 
    is a two-dimensional, banana shaped target, shown in Figure \ref{fig:Banana}.
    \item Ill conditioned Gaussian 
    is as 100-dimensional Gaussian with randomly oriented covariance matrix and eigenvalues drawn from the gamma distribution. Its condition number is around $10^5$. 
    \item German credit sparse logistic regression (GC) 
    is a 51-dimensional Bayesian hierarchical model, where logistic regression is used to model the approval of the credit based on the information about the applicant. 
    \item Brownian Motion (Brownian)
    is a 32-dimensional hierarchical problem, where Brownian motion with unknown innovation noise is fitted to the noisy and partially missing data. 
    \item Item Response theory (IRT)
    is a 501-dimensional hierarchical problem where students' ability is inferred, given the test results.
    \item Stochastic Volatility (SV)
    is a 2519-dimensional hierarchical non-Gaussian random walk fit to the S\&P500 returns data. 
\end{itemize}

We follow \citet{hoffman_tuning-free_2022} and define the bias of the expectation value $\mathbb{E}[f]$ as
\begin{equation} \label{eq:bias}
    b^2_t[f]= \frac{(\mathbb{E}_{\rho_t}[f] - \mathbb{E}_{p}[f])^2}{\mathrm{Var}_p[f]},
\end{equation}
where ground truth expectation values $\mathbb{E}_{p}[f]$ and $\mathrm{Var}_p[f] = \mathbb{E}_{p} [(f(\x) - \mathbb{E}_{p}[f])^2]$ are computed analytically for the Ill Conditioned Gaussian and by very long sequential NUTS runs for the other targets.
As a measure of convergence we will either take the largest or the average second-moment bias across parameters:
\begin{equation}\label{eq: def max avg}
    b^2_{\mathrm{max}, t} \equiv \max_{1 \leq i \leq d} b^2_t[x_i^2] \qquad b^2_{\mathrm{avg}, t} \equiv \frac{1}{d} \sum_{i = 1}^{d} b^2_t [x_i^2] .
\end{equation}
Both biases are of interest in applications and have been used as the posterior convergence metric. We use the max definition to directly compare with the ensemble algorithms from \citep{hoffman_tuning-free_2022} and the average definition to directly compare with sequential algorithms from \citet{robnik_microcanonical_2024}.

\begin{figure}
    \centering
    \includegraphics[width=\linewidth]{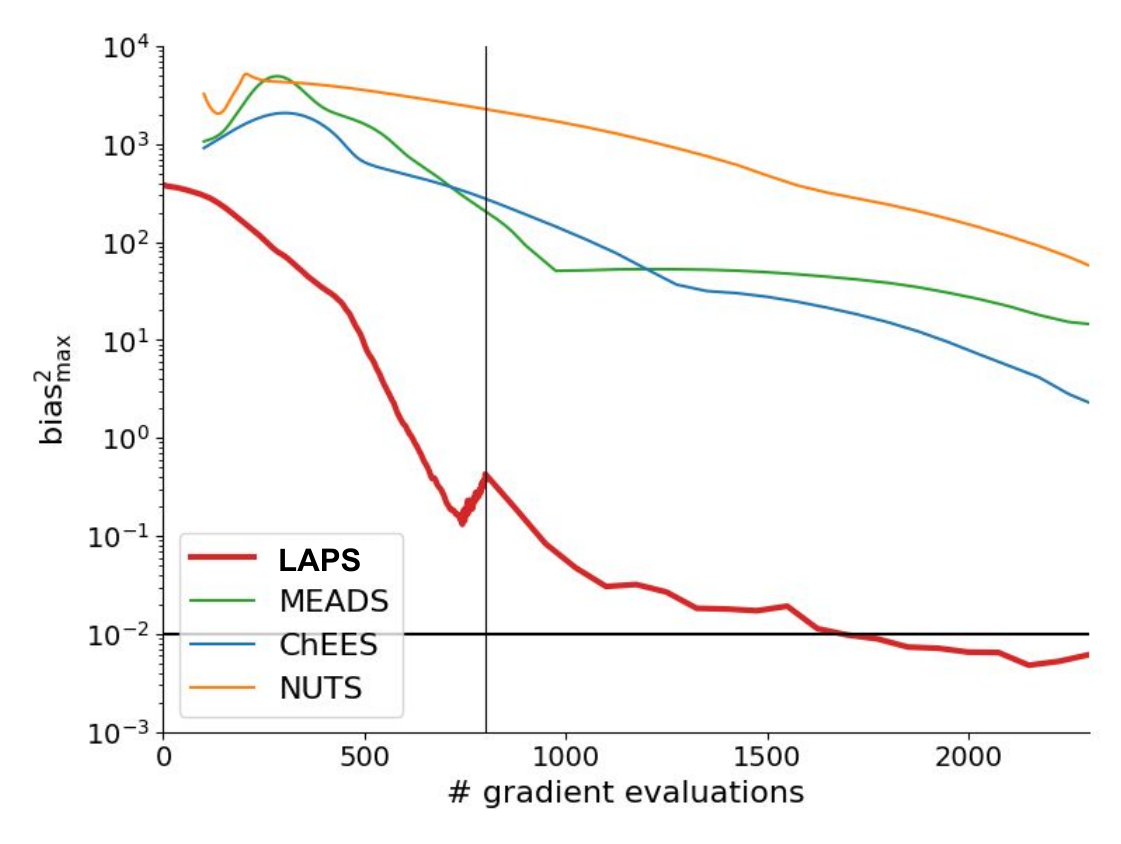}
    \caption{Second-moment bias as a function of the number of gradient evaluations on the Stochastic Volatility problem. LAPS converges significantly faster than the state-of-the-art ensemble samplers, which are shown for comparison.
    Note that in this case, the unadjusted phase does not quite achieve the desired accuracy requirements and the adjusted phase is essential for efficient convergence.}
    \label{fig:StochasticVolatility}
\end{figure}

\begin{table}[t]
\caption{Grads to low max bias ($b_{\mathrm{max}}^2 = 0.01$), normalized per chain, lower is better. LAPS is compared to the state-of-the-art ensemble algorithms. LAPS is the best performing algorithm in all cases, by factors between 2 to 20.}
\label{table}
\vskip 0.15in
\begin{center}
\begin{small}
\begin{sc}
\begin{tabular}{lccccr}
    \toprule & NUTS & ChEES & MEADS & LAPS \\\midrule
    Banana & 320 & 264 & 390 & {\bf 17}\\ 
    Gaussian & 9846 & 5138 & 5520 & {\bf 308}\\ 
    GC & 2716 & 1168 & 610 & {\bf 300}\\
    Brownian & 2159 & 600 & 410 & {\bf 206}\\
    IRT & 1162 & 537 & 790 & {\bf 185} \\
    SV & 4173 & 3080 & 2860$^\star$ & {\bf 1325}\\
\bottomrule
\end{tabular}
\end{sc}
\end{small}
\end{center}
\vskip -0.02in
\footnotesize{$^\star$ This is the number reported in Table 2 of \citet{hoffman_tuning-free_2022}. It is inconsistent with their Figure 7, which suggests it should be around 4200. Our Figure \ref{fig:StochasticVolatility} is based on their Figure 7.}
\end{table}

\begin{table}[t]
\caption{Grads to low average bias ($b_{\mathrm{avg}}^2 = 0.01$). LAPS with 256 chains is compared to the state-of-the-art sequential algorithms. LAPS (total) counts gradient calls from all the chains, LAPS counts gradient calls per chain. If the computation is not parallelized across the chains, the runtime is determined by the total number of gradient calls and the sequential algorithms are typically faster. 
However, if the computation is parallelized, the number of gradient calls per chain determines the runtime and LAPS is significantly faster.
}
\label{table: sequential}
\vskip 0.15in
\begin{center}
\begin{small}
\begin{sc}
\begin{tabular}{lccccr}
  \toprule & NUTS & MCLMC & LAPS & LAPS \\
  &(seq.) & (seq.) & (total) & \\
  \midrule
    Banana & 64,200 & 5000 & 4400 & 17 \\ 
    Gaussian & 113,900 & 33,500& 68,600 & 270\\ 
    GC  & 10,400 & 4200 & 28,700 & 110 \\
    Brownian & 5400 & 1700 & 17,900 & 70 \\
    IRT & 7300 & 1600 & 21,000 & 80 \\
    SV  & 29,900 & 10,000 & 281,600 & 1100 \\
\bottomrule
\end{tabular}
\end{sc}
\end{small}
\end{center}
\vskip -0.1in
\end{table}
In typical applications, computing the gradients $\nabla \log p(\x)$ dominates the total sampling cost, so we take the number of gradient evaluations as a proxy of wall-clock time. 
As in \citep{hoffman_tuning-free_2022}, we measure a sampler's performance as the number of gradient calls $n$ needed to achieve low error, $b^2_{\mathrm{max}} < 0.01$. 

In table \ref{table}, we compare LAPS to NUTS, ChEES-HMC, and MEADS as implemented in \citet{hoffman_tuning-free_2022}. Being MH adjusted, these samplers initially converge slowly, which these address by first running ADAM for 100 steps \citep{hoffman_tuning-free_2022}. We add these 100 gradient calls to the total cost of sampling in Table \ref{table} and Figure \ref{fig:StochasticVolatility}. Note that ADAM initialization can be very far from the truth for hierarchical parameters. All samplers in this test use $M = 4096$ chains. In Appendix \ref{sec: ablations} we show that LAPS performance does not degrade if a smaller number of chains is used, as long as the number of chains is larger than $128-256$.
 
\begin{figure*}
    \centering
    \includegraphics[width=0.8\linewidth]{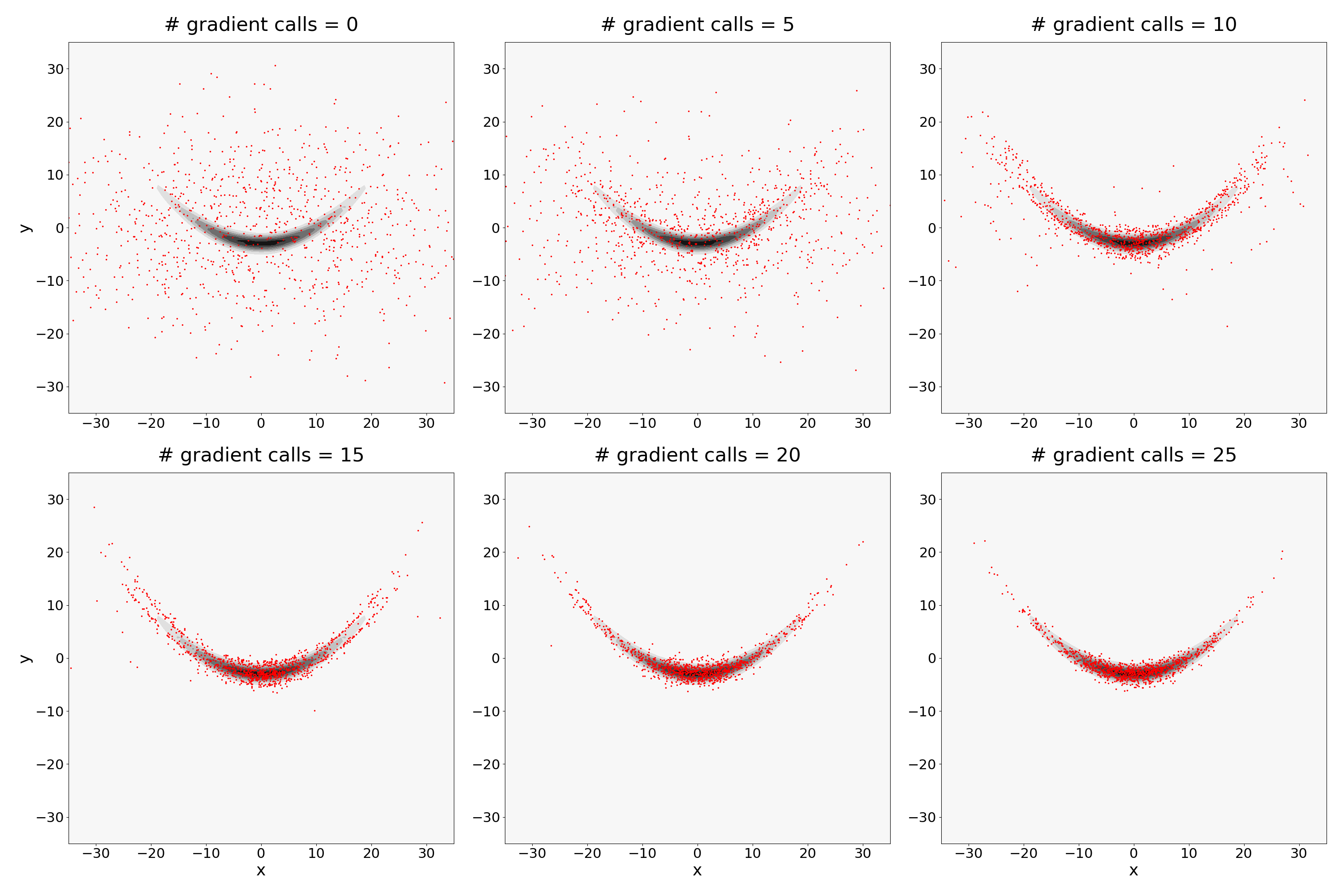}
    \caption{LAPS posterior on the Banana target, as a function of the number of gradient calls (per chain). The location of chains is shown in red, the target density in black. As can be seen, LAPS converges to the posterior in just around 20 gradient calls.}
    \label{fig:Banana}
\end{figure*}

LAPS outperforms other samplers on all benchmark problems, in some cases by an order of magnitude.
Convergence of hyperparameters in LAPS is shown in Figure \ref{fig:GermanCredit} and \ref{fig:StochasticVolatility}, where other samplers also also shown for comparison.
We also compared against the ensemble version of Pathfinder (Algorithm 2 in \citet{zhang_pathfinder_2022}). 
Pathfinder failed to converge for most of our benchmark problems. It only converged for the 2d Banana problem, where it reached
$b^2_{\mathrm{max}} = 0.02$ in 67 gradient evaluations, and for the German Credit problem, where it reached $b^2_{\mathrm{max}} = 21.5$ in 135 gradient evaluations, as shown in Figure \ref{fig:GermanCredit}. Thus, even when it does not fail completely, the unadjusted phase of LAPS obtains a better approximation to the posterior significantly faster.

We also compare LAPS to sequential algorithms, NUTS, and MCLMC. Two quantities of interest are the total number of gradient calls to low bias (which is proportional to the total number of CPU/GPU hours) and the number of gradient calls per chain (which is proportional to the wall-clock time). In this setting, we use $M=256$ chains for LAPS. The results are shown in Table \ref{table: sequential}. In terms of the total number of gradient calls, sequential algorithms are faster, but in some cases the difference is not large, demonstrating that in these cases MCLMC cold start burn-in cost is comparable to the cost of achieving an independent sample from a warm start. However, in terms of wall-clock time, LAPS is typically one to two orders of magnitude faster than sequential MCLMC and two to three orders of magnitude faster than NUTS.

\section{Conclusions}

Parallel MCMC is gaining popularity in the recent years due to the increased availability of parallel resources such as CPU clusters, GPUs and TPUs. It is especially attractive when likelihood evaluations are time consuming, as it can significantly reduce the wall-clock time. However, its cost can easily be dominated by the speed of convergence from the cold start.

We have formulated the Late Adjusted Parallel Sampler, a scheme that initializes sampling with an unadjusted sampler. The key challenge of unadjusted sampling is determining the step size that ensures small asymptotic bias. We have argued that the optimal step size should be time dependent: initially it should be large and become smaller as the chains approach the target distribution and better precision is needed. We have devised a scheme that achieves this by monitoring the equipartition condition, as a proxy for convergence, and appropriately adapting the step size. 
We have demonstrated on standard benchmark problems that the resulting algorithm outperforms the state of the art parallel samplers by up to an order of magnitude.

\paragraph{Robustness}
Construction of the unadjusted phase of the algorithm is inspired by the results that are largely based on the Gaussian assumption, which is a clear limitation.
For example, it is possible to construct pathological cases where $\widetilde{D}(p, \rho_t)$
is zero (or very small) but the distributions are not close. This would cause the algorithm to use a step size that is too short and thus progress slowly. On the other hand we have observed that for the distributions with long tails (such as the Banana example) the $\widetilde{D}$
is an overly stringent metric and even for exact i.i.d. samples decays only slowly with the number of samples. This causes the step size to become too large and the ensemble gets stuck, limited by the discretization bias. The result of both of these failure modes is that the ensemble reaches a stationary distribution which is not equal to the target distribution. However, since the expectation values settle at this point in both cases, the adjusted phase, which does not rely on these assumptions is triggered. We find that this makes the scheme robust, i.e. in the worst case, the adjusted phase is triggered.

Similarly, EEVPD $\propto \epsilon^6$ is not required to hold exactly for the algorithm to work, important is the feedback loop which makes sure that the step size is reduced if the EEVPD is larger than expected and to increase if it is smaller than expected, so the more broadly valid assumption of monotonic growth of EEVPD with the stepsize is sufficient. The upper left panel of Figure \ref{fig:GermanCredit} shows that this feedback loop works well for a non-Gaussian distribution and demonstrates that the observed EEVPD (orange) tracks its desired value (grey). In general, we have tested the algorithm on various non-Gaussian distributions and the performance and the stability of quantities like the EEVPD indicate that the algorithm is robust to the assumptions.

\paragraph{Future work} 
Several parallel methods such as sequential Monte Carlo, annealing and parallel tempering approach the problem of multi-modal distributions. These approaches are agnostic with respect to the kernel being used, so they are compatible with LAPS. For example in annealing, one could run LAPS at each temperature level, to bring the distribution of particles from the previous temperature level to the new stationary distribution as quickly as possible.

Equipartition was used in LAPS as a convergence metric, but this is not the only option. One alternative is the split-$\hat{R}$ statistics \citep{margossian_nested_2024}, which could also be useful for determining when to switch on adjustment. 

Finally, extensive testing on a variety of real-world problems with heavy tails, multimodality or complex geometry will help establish the LAPS strategy in the broader context.



\subsection*{Acknowledgements}
We thank Reuben Cohn-Gordon for many discussions, comments on the manuscript and help with Blackjax. We thank Nawaf Bou-Rabee, Bob Carpenter and Siddharth Mitra for useful discussions. 
This material is based upon work supported in part by the Heising-Simons Foundation grant 2021-
3282 and by NSF CSSI grant award number 2311559.

\bibliography{references,references2}

\appendix
\thispagestyle{empty}

\onecolumn
\aistatstitle{Supplementary Materials}

\section{Optimal step size schedule} \label{sec: schedule}

In this section we provide theoretical guidance for determining the step size schedule in the unadjusted phase of the algorithm. We here assume that the trajectory length $L$ is fixed and will only consider the schedule for the step size $\epsilon$. Let's denote by $K_{\epsilon}$ the transition operator corresponding to the Markov kernel $\varphi(\cdot \vert \epsilon)$. We have dropped the dependence on $L$ for shorter notation. Let's denote the stationary distribution of $K_{\epsilon}$ by $p_{\epsilon}$, that is, $K_{\epsilon} \, p_{\epsilon} = p_{\epsilon}$. We construct the Markov chain by applying the Markov kernel with step size $\epsilon_{t}$ in iteration $t$, such that the chain distribution evolves as $\rho_{t+1} = K_{\epsilon_{t}} \rho_t$.
The core of the ECA framework is that $\epsilon_{t}$ may be informed by $\rho_t$.

We will quantify convergence as a distance $D(\rho_t, p)$ between the chain distribution and the target distribution. Here, $D$ is a metric on the space of distributions, for example, the Wasserstein metric or the total variation distance. Since each kernel $K_{\epsilon_t}$ application performs $T / \epsilon_t$ steps, we define the per-step rate of convergence as
\begin{equation}
    R_t = \frac{D(\rho_t, p) - D(\rho_{t+1}, p)}{T / \epsilon_t}.
\end{equation}

In this section we will make the following assumptions about the Markov kernel and its asymptotic bias:
\begin{itemize}
    \item \textbf{Assumption 1.} Asymptotic bias can be bounded by
    \begin{equation} \label{A1}
        D(p, p_{\epsilon}) \leq c_{asym} \epsilon^{\kappa} \tag{A1},
    \end{equation}
    where the constant $c_{asym} > 0$ and order $\kappa > 0$ do not depend on the step size.

    \item \textbf{Assumption 2.} Dynamics is contracting towards the biased limit:
    \begin{equation} \label{A2}
        D(K_{\epsilon} \rho, p_{\epsilon}) \leq (1-c) D(\rho, p_{\epsilon}), \tag{A2}
    \end{equation}
    for any distribution $\rho$ and step size $\epsilon$. Contraction constant $0 < c < 1$ does not depend on the step size.
\end{itemize}

For example, for unadjusted HMC with the leapfrog integrator and the target distribution that is log-convex ($ (-\nabla \log p(\x) + \nabla \log p(\boldsymbol{y})) \cdot (\x - \boldsymbol{y}) \geq \Vert \x - \boldsymbol{y} \Vert^2$) and has bounded second and third derivatives ($\Vert \nabla^2 \log p(\x)\Vert \leq M_2$ and $\Vert \nabla^3 \log p(\x) \Vert \leq M_3 $), these assumptions hold with $\kappa = 2$, $c_{asym} = C(M_2, M_3, L) D / c$, $D = d + \int \Vert \x \Vert p(\x) d\x + \int \Vert \x \Vert^2 p(\x) d\x$ and $c = m L^2 / 4$.

We have the following result:
\begin{theorem}
    Given Assumptions \ref{A1} and \ref{A2} and a step size schedule, which is recursively constructed in such a way that it satisfies
    \begin{equation} \label{eq: schedule}
        c_{asym} \epsilon_{t}^{\kappa} = \frac{c}{(2-c)(\kappa+1)} D(\rho_t, p),
    \end{equation}
    the chain is guaranteed to converge to the target distribution $p$:
    \begin{equation}\label{eq: rate}
        D(\rho_{t+1}, p) \leq  (1 - c \frac{\kappa}{\kappa+1}) D(\rho_t, p) .
    \end{equation}
    This schedule is optimal in the sense that it maximizes the lower bound for the rate of convergence $R_t$.
\end{theorem}
\begin{proof}
Using Assumption \ref{A2} and applying triangle inequality twice gives
\begin{align*}
    D(\rho_{t+1}, p) &\leq D(\rho_{t+1}, p_{\epsilon}) + D(p, p_{\epsilon}) \\ &\leq (1-c) D(\rho_t, p_{\epsilon}) + D(p, p_{\epsilon}) \\ &\leq (1-c) D(\rho_t, p) + (2-c) D(p, p_{\epsilon}).
\end{align*}
We note that a sufficient condition for the convergent chain is $D(p, p_{\epsilon}) < \frac{c}{2-c} D(\rho_t, p)$. Applying Assumption \ref{A1} and inserting the step size schedule \eqref{eq: schedule} gives
\begin{align*}
    D(\rho_{t+1}, p) &\leq (1-c) D(\rho_t, p) + (2-c) c_{asym} \epsilon_t^{\kappa} \\
    &\leq (1 - c \frac{\kappa}{\kappa+1}) D(\rho_t, p).
\end{align*}
The rate of convergence is lower bounded by
\begin{equation*}
    R_t \geq \frac{\epsilon_t}{T} \big( c D(\rho_t, p) - (2-c) c_{asym}\epsilon_t^{\kappa} \big).
\end{equation*}
The lower bound is maximized for the unique $\epsilon_t$ for which
\begin{equation*}
    \frac{\partial R_t}{ \partial \epsilon_t} = \frac{1}{T} \big( c D(\rho_t, p) + (2-c) (\kappa+1) c_{asym} \epsilon_t^{\kappa}\big) = 0,
\end{equation*}
which gives the schedule \eqref{eq: schedule}. 
\end{proof}

This result suggests that the optimal step size schedule should control the asymptotic bias to be by a constant factor smaller than the total bias.

\paragraph{Numerical illustration} In Figure \ref{fig: stepsize} we compare the adaptive step size used in LAPS with the fixed step size strategy, all other elements being equal. It can be seen that the adaptive step size schedule outperforms any fixed given fixed step size: large steps are initially fast, but soon reach the asymptotic bias, while small step sizes are slow to converge.

\begin{figure}
    \centering
    \includegraphics[width=0.8\linewidth]{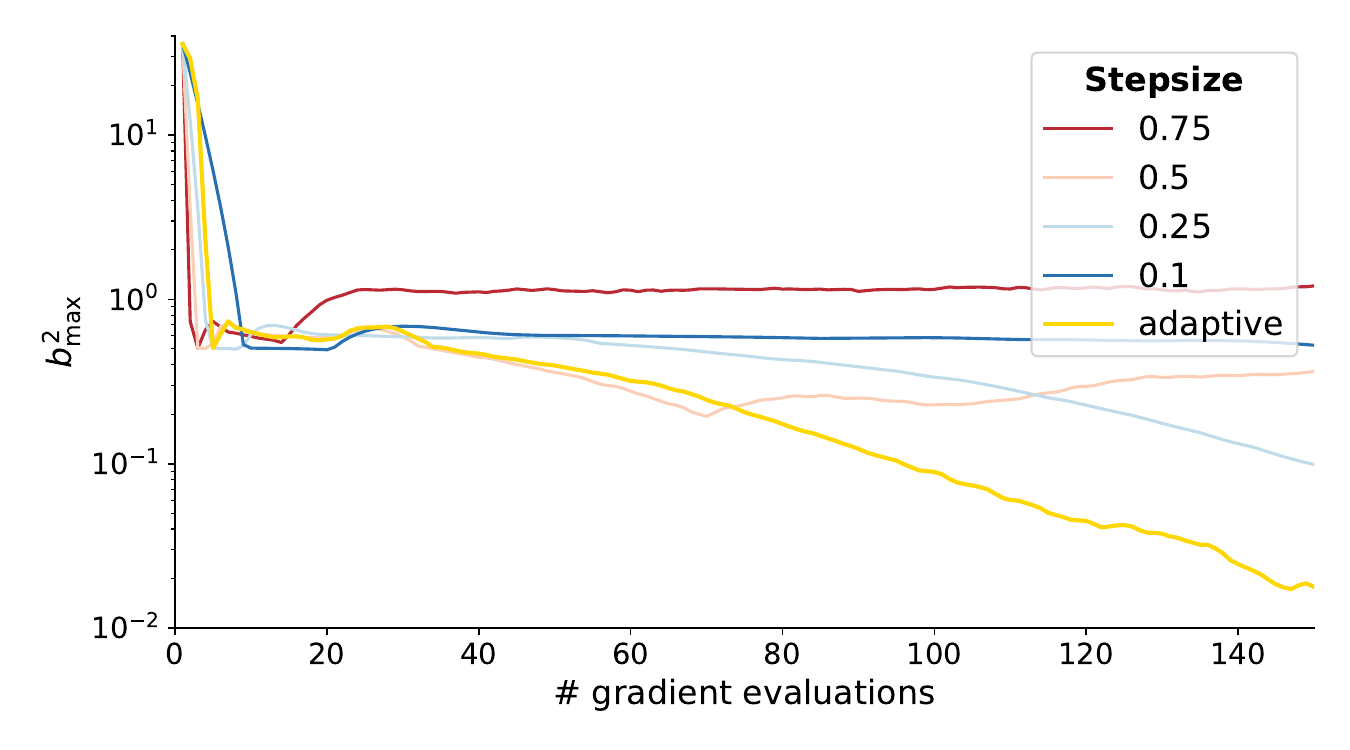}
    \caption{Bias as a function of the number of gradient evaluations with the unadjusted MCLMC dynamics for the German Credit problem. Adaptive step size adaptation used in LAPS (yellow) is compared with a fixed step size strategy for various step size choices. Large step size strategy (red) is initially very fast, but is not able to achieve fine convergence. Small step size (blue) is slow to converge but has a smaller asymptotic bias. LAPS adaptively reduces the step size, combining the best of all worlds.}
    \label{fig: stepsize}
\end{figure}

\section{Equipartition conditon} \label{sec: equipartition}

\begin{figure}[t]
    \centering
    \includegraphics[width=\linewidth]{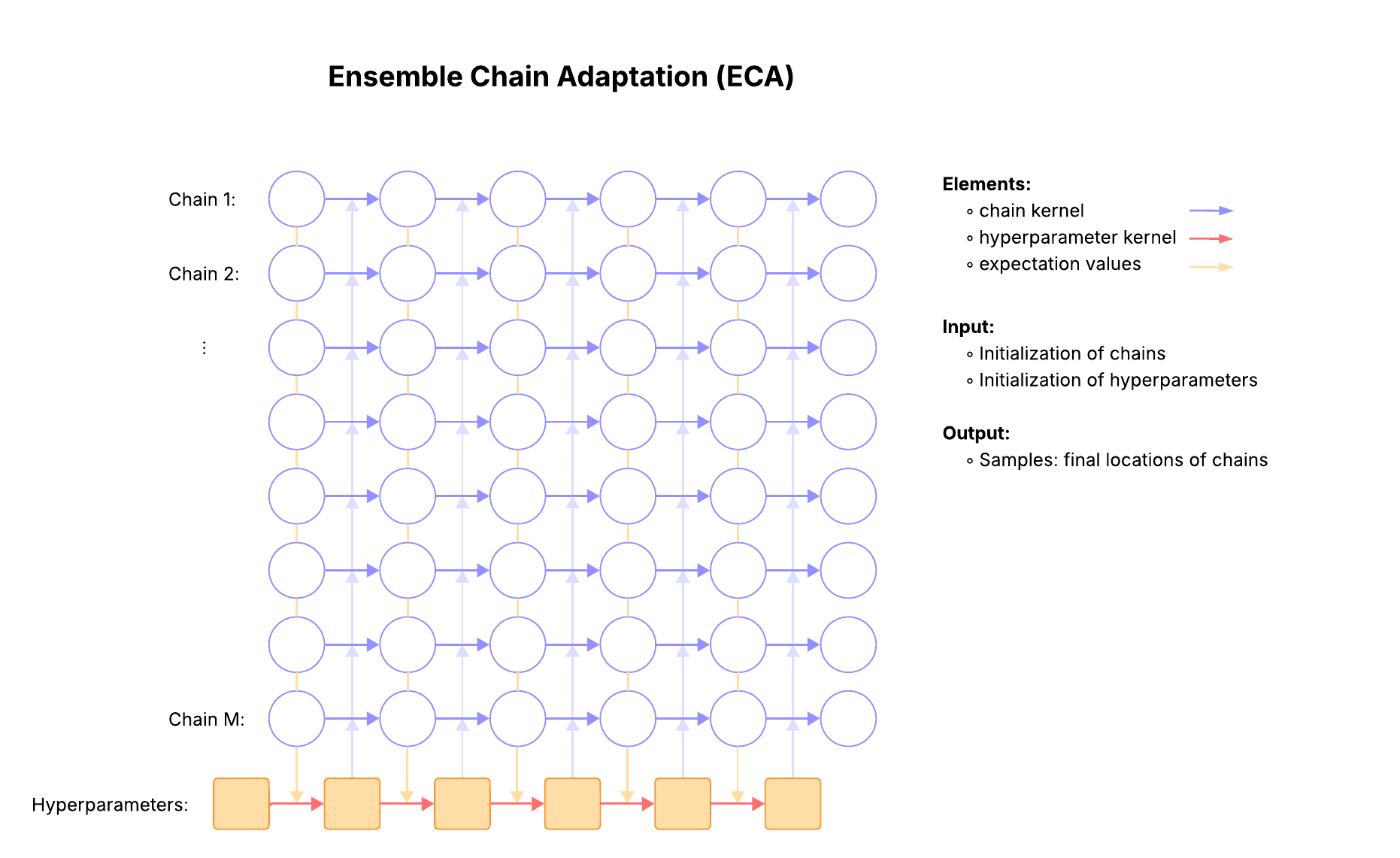}
    \caption{Schematic of the Ensemble Chain Adaptation (ECA) algorithm that we adopt. $M$ chains are evolved, time goes in the horizontal direction. In each step, the locations of all the chains are used to compute expectation values (orange arrow) which are used to update the hyperparameters (red arrow). Updated hyperparameters specify the kernel (pale purple arrow) which updates the chains in parallel (solid purple arrows).}
    \label{fig:eca}
\end{figure}
\subsection{Computation}

\paragraph{Full rank.} Naively computing $\widetilde{D}$ from Equation \eqref{eq: equipartition} requires $\mathcal{O}(d^2 M)$ operations or $\mathcal{O}(d M^2)$ by some rearranging of the sums. This is prohibitively expensive. However, it can be reduced to $\mathcal{O}(d M)$ (so $\mathcal{O}(d)$ per chain) by the Hutchinson's trick \citep{hutchinson1989stochastic}: a vector of random variables $\Z$ with $\mathbb{E} [z_i z_j ] = \delta_{ij}$ is introduced and inserted it in the trace to obtain
\begin{equation}
    \widetilde{D} = \frac{1}{d} \sum_{i, k, l = 1}^{d} \mathbb{E}[ (I-E)_{ik} z_k (I-E)_{i l} z_l ] \approx \mathbb{E}\bigg{[} \bigg{\vert} \frac{1}{d M}\sum_{m = 1}^M \boldsymbol{z} -  (\grad \mathcal{L}(\x^{m}) \cdot \boldsymbol{z}) \x^{m}  \bigg{\vert}^2 \bigg{]},
\end{equation}
where the ECA approximation \eqref{eq: ECA} was used in the second step (which is bias-free and with variance decaying as $1/M$). Two popular choices for the $\Z$ distribution are Gaussian distribution $\Z \sim \mathcal{N}(0, I)$ and the Rademacher distribution, we pick the latter.
We have checked that using 100 realizations of the $\Z$ variable ensures sub-percent error in the $\widetilde{D}$ computation.

\paragraph{Diagonal.} Run-time overhead of the above calculation is negligible, for all but the simplest problems, because it only requires minimal interaction between the chains (computing 100 expectation values). This is typically negligible compared to the gradient evaluation cost. However, it does introduce the cost of storing 100 $d$-dimensional $\boldsymbol{z}$ vectors. In some applications this is prohibitive, as the model computation already takes the entire GPU. We therefore also consider a simplification of $\widetilde{D}$ which only considers the diagonal terms of the equipartition matrix:
\begin{equation}
    \widetilde{D}_{\mathrm{diag}} = \frac{1}{d} \sum_{i = 1}^d (1 - V_{ii})^2.
\end{equation}
Evaluating $\widetilde{D}_{\mathrm{diag}}$ does not introduce the memory overhead. We will therefore use $\widetilde{D}_{\mathrm{diag}}$ as default. In Table \ref{table fullrank} we show that the performance does not degrade relative to the full rank equipartition matrix.

\subsection{Divergence property}

We will here show that $\widetilde{D}(p, p')$ is a divergence on the space of zero-mean Gaussian distributions, meaning that it is non-negative and positive if and only if $p = p'$. Similarly $\widetilde{D}_{\mathrm{diag}}$ is a divergence on the space of zero-mean uncorrelated Gaussian distributions. Zero-mean Gaussian distributions are completely determined by their covariance matrix $\Sigma$, we will denote them by $\mathcal{N}(\x \vert \Sigma)$. The equipartition matrix in this case simplifies to:
\begin{align} \nonumber
    V_{ij}(\mathcal{N}, \mathcal{N}') = \int x_i \partial_j \big(\frac{1}{2} \x^T \Sigma^{-1} \x \big) \mathcal{N}(\x \vert \Sigma') d \x
    = [\Sigma' \Sigma^{-1}]_{ij},
\end{align}
meaning that
\begin{align}
    \widetilde{D}(\mathcal{N}, \mathcal{N}') &= \frac{1}{d} \Tr{(I - \Sigma' \Sigma^{-1})(I - \Sigma' \Sigma^{-1})^T} = \frac{1}{d} \Vert I - \Sigma' \Sigma^{-1}\Vert_F^2\\ \label{eq:D2}
    \widetilde{D}_{\mathrm{diag}}(\mathcal{N}, \mathcal{N}') &= \frac{1}{d} \sum_{i=1}^d (1 - \Sigma'_{ii} / \Sigma_{ii})^2,
\end{align}
where $\Vert \cdot \Vert_F$ is the Frobenious norm.
$\widetilde{D}(\mathcal{N}, \mathcal{N}')$ is evidently positive (because the norm is positive) and zero if and only if $I - \Sigma' \Sigma^{-1} = 0$, which is equivalent to $\Sigma = \Sigma'$. Covariance matrix uniquely parameterizes the zero-mean Gaussian distribution so this demonstrates that $\widetilde{D}$ is a divergence on this space.
Similarly Equation \eqref{eq:D2} is a sum of positive terms, so it is zero if and only if all of the terms are zero, which is equivalent to $\Sigma'_{ii} = \Sigma_{ii}$ for all $i$. Thus, $\widetilde{D}_{\mathrm{diag}}$ is a divergence on the space of uncorrlated Gaussian distributions, because
$\{ \Sigma_{ii} \}_{i=1}^d$ uniquely determine the uncorrelated Gaussian distributions.

We note that Equation \eqref{eq:D2} coincides with the covariance matrix error as defined in \citep{robnik_black-box_2024}.


\newpage
\section{LAPS details} \label{sec: alg}

We here provide some further details and clarifications of the LAPS algorithm. 

\paragraph{Graphical illustration.} Hyperparameters in LAPS are updated by taking averages over the chains. This is illustrated in Figure \ref{fig:eca}.

\textbf{Pseudocode} is shown in Algorithm \ref{alg}. The first while loop corresponds to the unadjusted phase, and the second to the adjusted phase. The MCLMC dynamics updates for different chains are independent and can be executed in parallel. The hyperparameters are then updated using expectation values over the ensemble.

\begin{algorithm}
   \caption{LAPS Pseudocode}
   \label{alg}
\begin{algorithmic}
   \STATE {\bfseries Input:}
       \STATE differentiable unnormalized log-probability density $\log p(\x)$,
       \STATE initial particle positions $\{ \x^m\}_{m=1}^M$,
       \STATE initial particle velocities $\{ \U^m\}_{m=1}^M$
          
   \STATE {\bfseries Output:} samples $\{ \x^m\}_{m=1}^M$.
    \STATE
  \STATE $\epsilon \gets 0.01 \sqrt{d}$
   \STATE $L \gets$ Equation \eqref{eq: L} ($\{ \x^m\}$)
   \STATE $t \gets 0$
   \REPEAT
        \STATE $\x^m, \U^m, \Delta^m \gets $ MCLMC update \eqref{eq: discrete update}$(\x^m, \U^m, \epsilon, L)$
        \STATE $\mathrm{EEVPD} = \frac{1}{M} \sum_{i=1}^M (\Delta^m)^2 / d$ 
        \STATE $\widetilde{D} \gets$ Equation \eqref{eq: equipartition}$(\{ \x^m \})$
        \STATE $\mathrm{EEVPD}_{\mathrm{wanted}} \gets F(C \widetilde{D})$, see Equation \eqref{eq: bias bound}.
        \STATE $\epsilon \gets \epsilon \big(\mathrm{EEVPD}_{\mathrm{wanted}} /  \mathrm{EEVPD})^{1/6}$
       \STATE $L \gets$ Equation \eqref{eq: L} ($\{ \x^m\}$) 
       \STATE $\delta[x_i^2] \gets$ Equation \eqref{eq: relative fluctuations} ($\{ \x^m\}$) 
       \STATE $t \gets t+1$

    \UNTIL{$\max_i \delta[x_i^2] > 0.01$ and $t < \mathrm{maxiter}$}
   
    \STATE $\mathrm{ADAPT} \gets \mathrm{True}$
    \WHILE{$t < \mathrm{maxiter}$}
        \STATE $\x^m, \U^m, a^m \gets $ MAMS update  $(\x^m, \epsilon)$, see Section \ref{sec: MAMS}.
        \STATE $a = \frac{1}{M} \sum_{i=1}^M a^m$ 
        \STATE $\mathrm{ADAPT} \gets \mathrm{ADAPT} \, \wedge \vert a - a_{targeted}\vert > 0.03$
        \IF{ADAPT}
        \STATE $\epsilon \gets$ Bisection $(\epsilon, a - a_{\mathrm{targeted}})$
        \ENDIF
    \ENDWHILE
\end{algorithmic}
\end{algorithm}

\begin{figure}
    \centering
    \includegraphics[width=\linewidth]{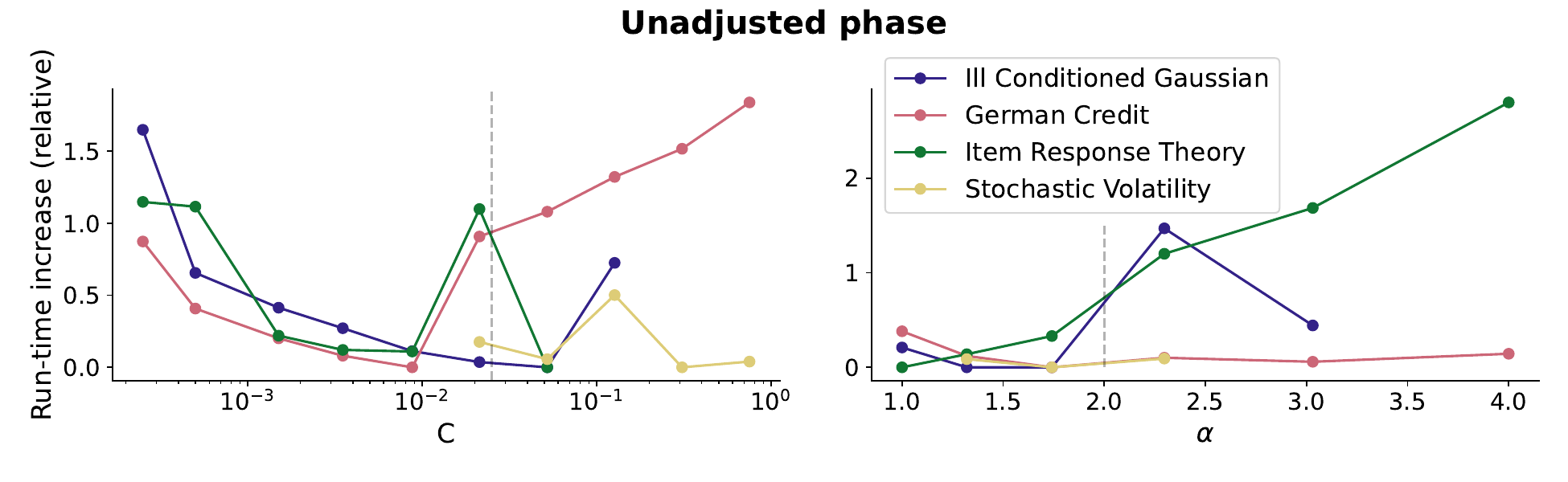}
    \caption{Ablation study of the unadjusted phase hyperparameters $C$ and $\alpha$. Run-time increase (as measured by the number of gradient calls to low bias, $b^2_{\mathrm{max}} < 0.01$) is measured relative to the optimal performance in this grid search. The default hyperparameter values are shown by vertical dashsed lines.}
    \label{fig:ablations1}
\end{figure}


\paragraph{Memory complexity} The memory complexity of the algorithm is $\mathcal{O}(d)$ per chain, where $d$ is the dimensionality of the parameter space on which the inference is performed. Complexity is independent of the number of steps.

\paragraph{Architecture.} The experiments were run on the NVIDIA A100 GPU (40GB).

\paragraph{Initialization} The initial positions of the chains are provided by the user. In the experiments in this work, we initialize them by drawing samples from the prior distribution, as defined by the Inference Gym models.
We initialize the velocity to be aligned with the gradient (and it has unit norm in MCLMC). 

\section{Ablation studies} \label{sec: ablations}

\begin{figure}
    \centering
    \includegraphics[width=\linewidth]{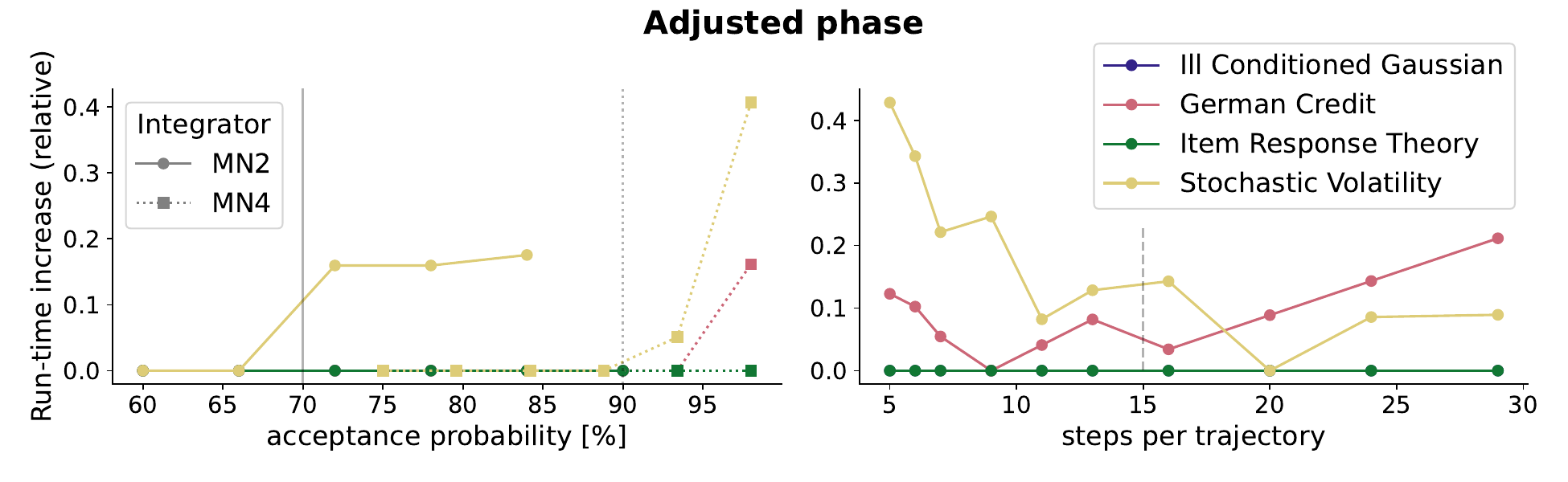}
    \caption{
    Ablation study of the adjusted phase hyperparameters acceptance probability and the number of steps per trajectory. Run-time increase (as measured by the number of gradient calls to low bias, $b^2_{\mathrm{max}} < 0.01$) is measured relative to the optimal performance in this grid search. The default hyperparameter values are shown by vertical dashsed lines. Acceptance rate results are also shown as a function of the integrator used (fourth order MN4 versus second order MN2). 
    }
    \label{fig:ablations2}
\end{figure}

\begin{figure}
    \centering
    \includegraphics[width=0.9\linewidth]{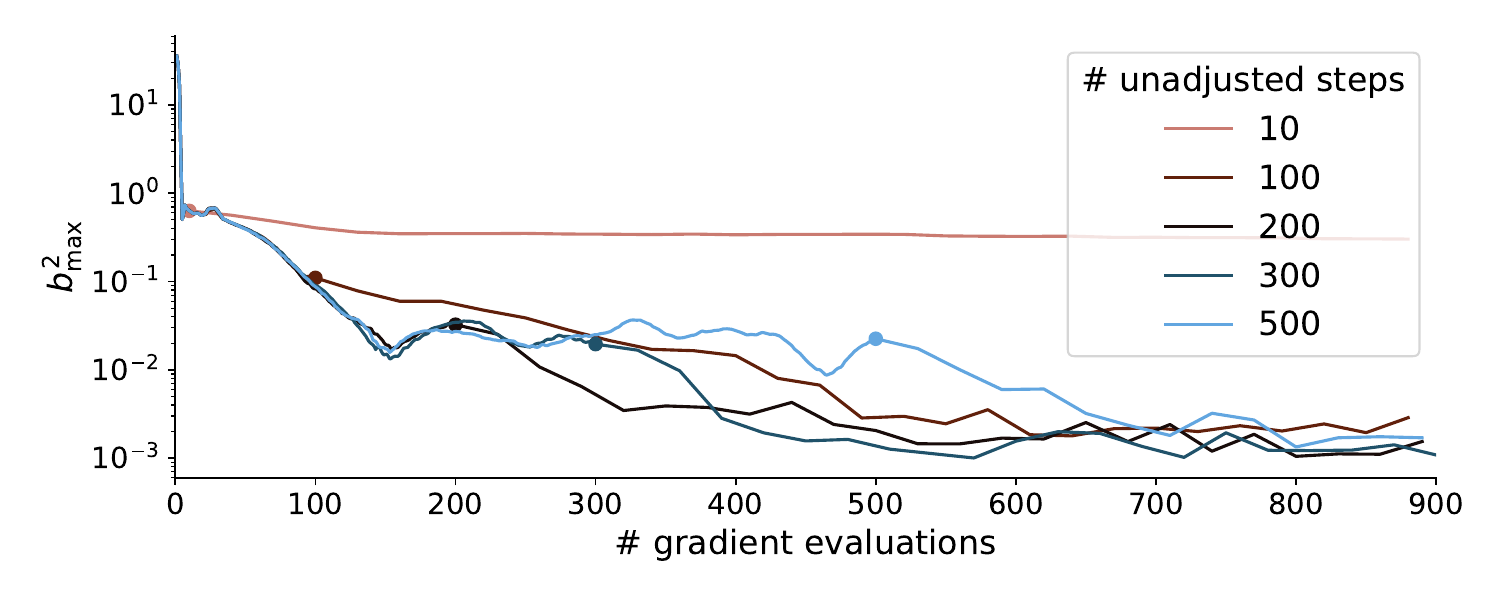}
    \caption{Dependence of the performance on the number of unadjusted steps. Bias as a function of the number of total gradient calls is shown for the German Credit problem. Different colors correspond to different number of steps in the unadjusted phase. The switch to adjustment is denoted by a dot. The adjusted method on its own (10 unadjusted steps) struggles with the cold start and does not converge in a reasonable time. If a sufficient number of unadjusted steps is taken (100) the method converges, but the adjusted part takes an unnecessarily long time. On the other hand the unadjusted method on its own does not converge to the extremely low bias, which is only achieved once the adjustment is turned on, as can be when 500 unadjusted steps are taken (light blue). In LAPS, the switch to adjustment is determined automatically by monitoring the summary statistics and in this case occurs after 200 gradient calls, which is close to optimal.}
    \label{fig:n1}
\end{figure}

\begin{figure}
    \centering
    \includegraphics[width=\linewidth]{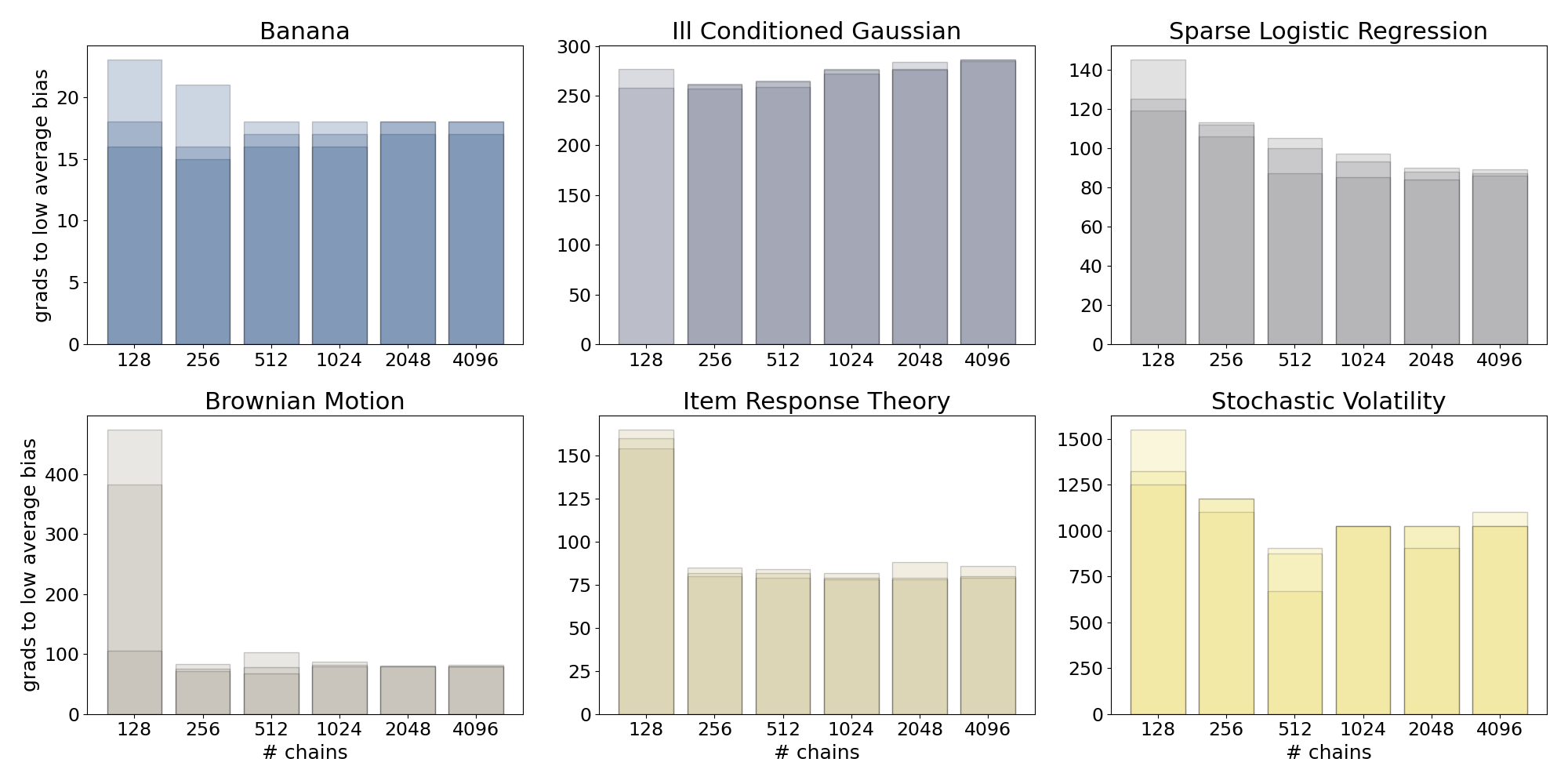}
    \caption{Number of gradient calls to low bias ($b_{\mathrm{avg}}^2 < 0.01$) as a function of number of chains $M$. Different color shades correspond to runs with different random seeds. It can be seen that performance stays roughly constant as long as $M > 128-256$, depending on the target.}
    \label{fig:chains}
\end{figure}

We here study the performance of the algorithm as a function of various algorithmic choices.

\paragraph{Unadjusted phase.} We study the performance as a function of $C$ (ratio between the asymptotic bias and the total bias) and $\alpha$ (the proportionality constant relating the typical posterior size and the momentum decoherence scale $L$). The results are shown in Figure \ref{fig:ablations1}. These results demonstrate that the performance is stable in a wide range of hyperparameters. Results for some of the models could be improved for some of the models (German Credit and Item Response Theory) by a factor of two if a smaller value of $C$ was used.

\paragraph{Adjusted phase.} We study the performance as a function of the acceptance probability and the number of steps per trajectory. The results are shown in Figure \ref{fig:ablations2}. The acceptance probability ablations are shown for both MN2 and MN4 integrators, demonstrating that the choice of the integrator does not play a large role. Figure \ref{fig:ablations2} shows that the performance is very stable with respect to the adjusted phase hyperparameters. For some models (Ill Conditioned Gaussian and Item Response Theory) the unadjusted phase already achieved close to desired accuracy, so these hyperparameters practically do not play a role. In all cases the change of performance in the reasonable range of hyperparameters is on the order of $10\%$.

\paragraph{Is the adjusted phase necessary?} We study how does the time when we switch on adjustment impact performance. This is shown in Figure \ref{fig:n1} for the German Credit problem. Turning on adjustment too soon dramatically slows down convergence, as expected. Unadjusted sampler on the other hand reaches a stationary state at some point and does not improve further. Adjusted method is needed for the final decay to extremely low error (determined by the variance due to the finite number of chains). Without the adjusted phase the sampler fails to achieve low error for some problems, for example for the Stochastic Volatility model.

\paragraph{Is the unadjusted phase necessary?} 
We compare LAPS with and without the unadjusted phase. The algorithm with only the adjusted phase is a parallel version of the Metropolis adjusted microcanonical dynamics (MAMS). We explore two options, with initializing the sampling with 100 ADAM steps (as in \citep{hoffman_tuning-free_2022}) or without. In both cases we give MAMS an (unfair) advantage of giving it the same diagonal mass matrix preconditioning that LAPS would use in the adjusted phase (i.e. the one obtained from samples after the LAPS unadjusted phase), so that we do not need to develop a custom adaptive preconditioner. We have checked that this produces better results than omitting the preconditioning and we expect it to be better than the custom preconditioner (which would be based on more biased samples in the initial steps). The results are shown in Table \ref{table: mams}. As can be seen, the unadjusted initialization (LAPS) gives a much more stable performance and is in many problems the only method that converges. When all methods converge, LAPS is typically significantly faster, the only exception is German Credit when performance of all methods is similar. For example, LAPS is a factor of 18-23 times faster than MAMS on Banana.
This results suggest that both adjusted and unadjusted phases are critical for the good performance.

\paragraph{Number of chains} We study the performance as a function of the number of chains $M$. The results are shown in Figure \ref{fig:chains}. The performance is roughly independent of the number of chains, as long as $M > 128-256$, depending on the target. Lower $M$ cannot work, because the number of samples obtained by LAPS is equal to $M$. Even if these are perfectly independent, unbiased samples, their accuracy is limited by the variance. $b_{\mathrm{avg}}^2 = 0.01$ accuracy roughly corresponds to 100 effective samples \citep{hoffman_tuning-free_2022}, so a lower number of chains typically cannot achieve this performance.

\paragraph{Equipartition estimation}
In the rest of the work we used the diagonal version of $\widetilde{D}$ as bias proxy in the unadjusted phase. Here we explore if this degrades the performance relative to the full rank $\widetilde{D}$. Table \ref{table fullrank} compares the the number of gradient calls to low bias ($b_{\mathrm{max}}^2 < 0.01$) for both approaches. It demonstrates that on the problems considered here, the differences are marginal. Note that this also holds for the Ill Conditioned Gaussian, whose covariance matrix is randomly orientated and therefore is not biased to be close to diagonal. This suggests that there is no need to use the full rank bias proxy.

\paragraph{Dynamics}
We tried replacing the microcanonical dynamics with the more standard underdamped Langevin dynamics. We found that the unadjusted underdamped Langevin dynamics fails to aid the sampler and the scheme is significanlty slower than LAPS with the microcanonical dynamics. For example, LAPS scheme with underdamped Langevin converges to
in 650 gradient calls per chain on the Item response theory. Only running the adjusted underdamped dynamics (i.e. MALT) converges in 610 gradient calls, so comparable. It is also comparable to HMC with ChESS (537 gradient calls) and underdamped Langevin with MEADS (790 gradient calls). LAPS with the microcanonical dynamics on the other hand converges significantly faster (185 gradient calls). The conclusion is that microcanonical dynamics is important for the success of our scheme.

\begin{table}
\caption{Grads to low max bias ($b_{\mathrm{max}}^2 = 0.01$), normalized per chain, lower is better. - indicates that sampler did not converge after 2000 gradient calls. LAPS is compared with only the adjusted phase, i.e., the parallel version of MAMS. MAMS is shown to be less reliable and often fails to converge in reasonable time.}
\label{table: mams}
\vskip 0.15in
\begin{center}
\begin{small}
\begin{sc}
\begin{tabular}{lcccr}
    \toprule & MAMS & MAMS with ADAM & LAPS \\\midrule
    Banana & 390 & 670 & 17 \\
    Ill conditioned Gaussian & - & - & 308 \\
    German Credit & 270 & 280 & 300 \\
    Item Response Theory & - & 1000 & 185 \\
    Stochastic Volatility & - & - & 1325 \\
\bottomrule
\end{tabular}
\end{sc}
\end{small}
\end{center}
\vskip -0.02in
\footnotesize{}
\end{table}

\begin{table}
\caption{Ratio of the number of gradient calls to low bias for the full rank equipartition versus the diagonal one. The differences are marginal, suggesting that diagonal version, which is less memory intensive, may be used.
}
\label{table fullrank}
\vskip 0.15in
\begin{center}
\begin{small}
\begin{sc}
\begin{tabular}{lcccccr}
    \toprule & Gaussian & German Credit & IRT & Stoch. Volatility\\\midrule
     Relative run-time (full rank / diagonal) & 1.003 & 0.954 & 1.102 & 1.000 \\
\bottomrule
\end{tabular}
\end{sc}
\end{small}
\end{center}
\end{table}

\section{Integrators} \label{sec: integrator}

In this appendix we review the numerical integrators used to approximate Equation \eqref{eq: SDE}. We employ operator splitting schemes that alternate between deterministic and stochastic updates \citep{leimkuhler_molecular_2015}. The resulting algorithms are time-reversible, stable, and suitable for use inside a Metropolis–Hastings wrapper \citep{robnik_metropolis_2025}.

First, we split the SDE in a stochastic and determinisitc part and approximate the combined solution by an Euler-Maruyama composition,
\begin{equation} \label{eq: discrete update}
   \varphi(\cdot \vert \epsilon, L) 
    = \Phi^O_{\epsilon/2, L} \circ \Phi_{\epsilon} \circ \Phi^O_{\epsilon/2, L}.
\end{equation}
Here, $\epsilon$ is the step size, $\Phi^O$ is the stochastic update, and $\Phi_{\epsilon}$ is the deterministic update.  

\begin{table}
\label{integrators}
\centering
\caption{Integrators considered in this work. The integrator is defined by its coefficients, which are given in the last column. The coefficients that can be inferred from the symmetric structure or normalization are not explicitly provided.}
\begin{tabular}{lcccc}
\toprule 
Name & Order & Gradients/step & Best for & Coefficients \\
\midrule 
LF & 2 & 1 & low accuracy & 
$b_1 = 0.5$ \\
\hline
MN2 & 2 & 2 & high accuracy, low dimensionality & 
   $\begin{aligned}
        b_1 &= 0.1931833275 \\
        a_1 &= 0.5
    \end{aligned}$
\\ \hline
MN4 & 4 & 5 & high accuracy, high dimensionality &
    $\begin{aligned}
        b_1 &= 0.0839831526 \\
        b_2 &= 0.6822365335 \\
        a_1 &= 0.2539785108 \\
        a_2 &= -0.032302867
    \end{aligned}$
\\
\bottomrule
\end{tabular}
\end{table}

The deterministic update is obtained by further splitting the differential equation to the part that updates the velocity and the part that updates the position. The exact solutions of these two equations are called the position update $\Phi^A_{\epsilon}$ and the velocity update $\Phi^B_{\epsilon}$. The combined solution is approximated by alternating these two updates:
\begin{equation} \label{eq: integrator}
    \Phi_{\epsilon} = 
    \Phi^B_{b_1 \epsilon}\,
    \Phi^A_{a_1 \epsilon}\,
    \Phi^B_{b_2 \epsilon}\,
    \Phi^A_{a_2 \epsilon}\,\cdots\,
    \Phi^B_{b_K \epsilon}.
\end{equation}
Different integrators correspond to different choices of coefficients \(\{a_k, b_k\}\). In all cases we require the normalization $\sum_{k=1}^K b_k = \sum_{k=1}^{K-1} a_k = 1$ and the symmetric composition (\(a_k = a_{K-k}\), \(b_k = b_{K+1-k}\)), because it ensures time-reversibility and cancels odd-order terms in the local error expansion, thereby improving accuracy.  
Specifically, we consider three schemes, summarized in Table \ref{integrators}. Each scheme trades accuracy for computational cost in a different way:
\begin{itemize}
    \item Leapfrog (LF; \citet{leimkuhler_molecular_2015}) requires only 1 gradient evaluation per step. It is preferred when low accuracy is sufficient, i.e. in the unadjusted phase of the algorithm.  
    \item Second order minimal norm integrator (MN2; \citet{omelyan_optimized_2002, omelyan_symplectic_2003}) is optimized for accuracy, and comes with two gradient evaluations per step. This makes it the preferred choice in the adjusted phase of the algorithm, when we are trying to achieve fine convergence.
    \item Fourth order minimal norm integrator (MN4; \citet{omelyan_optimized_2002, omelyan_symplectic_2003}) is more costly, as it requires five gradient calls per step. However, unlike the other two integrators, it is fourth order accurate, meaning that its local approximation error scales as $\epsilon^4$. This means that reducing the step size has a much larger impact on the accuracy than for the second order integrators. This is crucial for the adjusted phase if the target is high-dimensional \citep{takaishi_testing_2006, hernandez-sanchez_higher_2021}. We adopt it for $d > 200$.
\end{itemize}

For completeness, we below explicitly state the $\Phi^O$, $\Phi^A$ and $\Phi^B$ updats.

\paragraph{Stochastic update.} We follow \citep{robnik_black-box_2024} and write the stochastic update as 
\begin{equation} \label{eq: O update}
    \Phi^O_{\epsilon, L}(\x, \U) 
    = \left( \x, \,
    \frac{ c_1 \U + c_2 \boldsymbol{Z}}
         { \|c_1 \U + c_2 \boldsymbol{Z}\| } \right),
\end{equation}
where \(\boldsymbol{Z} \sim \mathcal{N}(0, I/\sqrt{d})\) is a $d$-dimensional vector, \(c_1 = e^{-\epsilon/L}\), and \(c_2 = \sqrt{1-c_1^2}\).

\paragraph{Position update.}  
At fixed velocity, the dynamics from Equation \eqref{eq: SDE} reduces to
\[
    \frac{d}{dt} \x_t = \U,
\]
which has the trivial solution
\begin{equation}
    \Phi^A_\epsilon: \quad (\x, \U) \mapsto (\x + \epsilon \U, \U).
\end{equation}
The energy change induced by this update is $\Delta(\x, \U \vert p, \epsilon) = - \log p(\x + \epsilon \U) + \log p(\x)$
\paragraph{Velocity update.}  
At fixed position, the deterministic part of the dynamics from Equation \eqref{eq: SDE} reduces to
\[
    \tfrac{d}{dt} \U_t = (I - \U_t \U_t^T)\,\frac{\nabla \log p(\x)}{d-1}.
\]
This is the nonlinear Riccati vector equation. However, it can be reduced to a second-order linear ODE by introducing the variable \(\boldsymbol{y}_t\) via \(\U_t = \tfrac{d}{dt}{\boldsymbol{y}_t} / (\boldsymbol{g}\cdot \boldsymbol{y}_t)\). One obtains \citep{steeg_hamiltonian_2021}
\[
    \tfrac{d^2}{dt^2}{\boldsymbol{y}} = (\boldsymbol{g} \boldsymbol{g}^T)\boldsymbol{y},
\]
which is a linear system with analytic solution. We have denoted $\boldsymbol{g} = \frac{\nabla \log p(\x)}{d-1}$ and $\boldsymbol{e} = \frac{\boldsymbol{g}}{\|\boldsymbol{g} \|}$.
Solving the linear equation and translating back to the $\U$ variable, one arrives at the closed-form update
\begin{equation} \label{eq: B MCHMC}
    \Phi^B_\epsilon: \quad 
    (\x, \U) \mapsto (\x, \frac{\U + \big(\sinh \delta + (\boldsymbol{e}\cdot \U)(\cosh \delta -1)\big)\boldsymbol{e}}
                      {\cosh \delta + (\boldsymbol{e}\cdot \U)\sinh \delta}),
\end{equation}
where \(\delta = \epsilon \|\nabla \log p(\x)\|/(d-1)\).  
Thus the velocity rotates smoothly in the plane spanned by \(\U\) and \(\boldsymbol{g}\), while maintaining unit norm.  
The energy change induced by this update is \citep{steeg_hamiltonian_2021}:
\begin{equation}
    \Delta(\x, \U \vert p, \epsilon) = (d-1) \log \Big\{ \cosh \delta + (\boldsymbol{e} \cdot \U)\sinh \delta \Big\}.
\end{equation}
    
\section{Metropolis adjusted microcanonical dynamics} \label{sec: MAMS}

In this appendix we review the Metropolis-adjusted microcanonical dynamics (MAMS; \citet{robnik_metropolis_2025}) that LAPS uses in the adjusted phase of the algorithm. 
MAMS kernel is constructed by the following steps:
\begin{enumerate}
    \item Start from $(\x, \U)$. 
    \item Resample $\U$ from a uniform distribution on the unit sphere. Specifically $\U = Z / \| Z\|$, with $Z \sim \mathcal{N}(0, I)$.
    \item Apply $N$ steps of the splitting integrator with step size $\epsilon$, including partial velocity refreshment at each step, to obtain $(\x', \U')$.  
    \item Compute the energy error $\Delta$ accumulated in this integration. Each position update $\x \xrightarrow[]{} \x'$ contributes $- \log p(\x') / p(\x)$ to the energy error. Each velocity update contributes
    $(d-1) \log \Big\{ \cosh \delta + (\boldsymbol{e} \cdot \U)\sinh \delta \Big\}$,
    where $\delta = \epsilon \|\nabla \log p(\x)\|/(d-1)$ and $\boldsymbol{e} = \nabla \log p(\x)/\|\nabla \log p(\x)\|$, see Section \ref{sec: integrator}.
    \item Accept the proposal with probability $\min(1, e^{-\Delta})$, otherwise remain at $(\x, \U)$.
\end{enumerate}
The stationary distribution of this kernel is exactly the target distribution, the integrator approximation error is eliminated by the   Metropolis–Hastings (MH) step.





\end{document}